\documentclass[a4paper]{article}
\usepackage{lmodern}
\pdfoutput=1
\usepackage{graphicx,wrapfig}
\usepackage{lipsum}
\usepackage{amsmath,amssymb,mathrsfs}
\usepackage{hyperref}
\usepackage{url}
\usepackage{mathtools}
\usepackage{changepage}
\usepackage[ruled,vlined]{algorithm2e}
\usepackage{amsmath}
\usepackage{multirow}
\usepackage{wrapfig}
\usepackage{subfig}
\usepackage{color}
\usepackage{epsfig,enumerate,amsmath,amsfonts,amssymb,amsthm,mathrsfs,ifpdf}
\usepackage{indentfirst,relsize}
\usepackage{setspace,graphicx}
\usepackage{latexsym}
\usepackage[all]{xy}
\usepackage[dvipsnames]{pstricks}
\usepackage{pst-grad} 
\usepackage{pst-plot} 
\usepackage[margin = 2.50cm]{geometry}

\usepackage{authblk} 
\usepackage{algpseudocode}
\usepackage{color,soul}
\usepackage{enumitem}

\usepackage{orcidlink}
\usepackage{tabularx}

\newcommand{\remove}[1]{}

\newtheorem{theorem}{Theorem}
\newtheorem{lemma}[theorem]{Lemma}

\usepackage{todonotes}
\usepackage[most]{tcolorbox}

\title{Roman-Type Domination on Convex and Chordal Bipartite Graphs: Algorithms and Hardness}

\author[1]{Gautam K. Das\footnote{gkd@iitg.ac.in}}
\author[1]{Kamal Santra \orcidlink{0009-0006-5997-1452} \footnote{kamal.7.2013@gmail.com, kamal.santra@iitg.ac.in}}
\affil[1]{Department of Mathematics\\
	
	Indian Institute of Technology Guwahati\\
	
	Guwahati, 781039, Assam, India}

\date{}

\begin{document}

	\maketitle
\begin{abstract}
	Roman domination and its variants form an important family of domination-type graph parameters motivated by protection, fault tolerance, and resource allocation. A Roman dominating function of a graph \(G\) is a function \(f:V(G)\rightarrow\{0,1,2\}\) such that every vertex \(v\) with \(f(v)=0\) has a neighbour \(u\) with \(f(u)=2\). The weight of \(f\) is \(w(f)=\sum_{v\in V(G)}f(v)\), and the minimum weight of a Roman dominating function of \(G\) is the Roman domination number, denoted by \(\gamma_R(G)\). In this paper, we study four variants of Roman domination on two natural subclasses of bipartite graphs, namely convex bipartite graphs and chordal bipartite graphs. On the positive side, we develop a unified left-to-right dynamic programming framework for Roman-\(\{2\}\) domination, double Roman domination, perfect Roman domination, and unique response Roman domination on convex bipartite graphs. The algorithms exploit the interval structure of one bipartition class and represent all unfinished requirements using a constant number of boundary indices. Consequently, each of the four parameters can be computed in \(O(n^6)\) time, where \(n=|V(G)|\). On the negative side, we prove that Roman-\(\{2\}\) domination, perfect Roman domination, and unique response Roman domination remain NP-complete on chordal bipartite graphs. These results establish a clear algorithmic separation between convex bipartite graphs, where the interval ordering yields polynomial-time solvability, and the broader class of chordal bipartite graphs, where several Roman-type domination problems remain computationally intractable.
\end{abstract}


\section{Introduction}\label{sec:introduction}

Domination is one of the central and most extensively studied topics in
graph theory and graph algorithms. A set \(D\subseteq V(G)\) is a
dominating set of a graph \(G\) if every vertex in \(V(G)\setminus D\) has
a neighbour in \(D\). Although the definition is simple, domination and its
numerous variants provide natural models for facility location, network
monitoring, communication systems, resource allocation, and protection
problems. We refer to the standard monographs
\cite{haynes2020topics,haynes1998advanced,haynes1998fundamentals} for a
comprehensive treatment of domination in graphs.

Roman domination is a weighted form of domination motivated by the problem
of deploying a limited number of military units to protect the regions of
the Roman Empire. The historical motivation was popularized by
Stewart~\cite{stewart1999defend} and discussed mathematically by ReVelle and
Rosing~\cite{revelle2000defendens}. The graph-theoretic notion was formally
introduced by Cockayne et al.~\cite{cockayne2004roman}. A Roman dominating
function, abbreviated as an RDF, of a graph \(G\) is a function
\(f:V(G)\rightarrow\{0,1,2\}\) such that every vertex assigned value \(0\)
has a neighbour assigned value \(2\). The weight of \(f\) is
\(w(f)=\sum_{v\in V(G)}f(v)\), and the minimum possible weight is the Roman
domination number \(\gamma_R(G)\).

The value \(2\) may be interpreted as placing two units of defence at a
vertex. One unit protects the vertex itself, while the second unit can be
moved to an adjacent undefended vertex when necessary. A vertex assigned
value \(1\) has enough defence only for itself, whereas a vertex assigned
value \(0\) must rely on a neighbouring vertex assigned value \(2\). Since
the introduction of the parameter, Roman domination has been studied from
structural, extremal, algorithmic, and complexity-theoretic viewpoints; see,
for example,
\cite{chambers2009extremal,liedloff2008efficient,liu2013roman,
	padamutham2020algorithmic}. Many variants have also been introduced by
changing the permitted values or strengthening the local protection
requirements.

In this paper, we consider four such variants: Roman-\(\{2\}\) domination,
double Roman domination, perfect Roman domination, and unique response Roman
domination. Although their local feasibility conditions are different, we
show that all four problems can be treated by a common left-to-right dynamic
programming framework on convex bipartite graphs.

\subsection{The Roman-type domination parameters}

\paragraph*{\textbf{Roman-\texorpdfstring{\(\{2\}\)} domination.}}

Roman-\(\{2\}\) domination was introduced by Chellali et
al.~\cite{chellali2016roman2}. A Roman-\(\{2\}\) dominating function of a
graph \(G\) is a function \(f:V(G)\rightarrow\{0,1,2\}\) such that every
vertex \(v\) with \(f(v)=0\) satisfies
\(\sum_{u\in N(v)}f(u)\geq 2\). Thus, a vertex assigned value \(0\) may be
protected either by one neighbour assigned value \(2\), or by at least two
neighbours assigned value \(1\). The minimum weight of such a function is
the Roman-\(\{2\}\) domination number, denoted by
\(\gamma_{\{R2\}}(G)\).

This parameter is also commonly studied under the name \emph{Italian
	domination}. In contrast with ordinary Roman domination, two vertices
carrying one unit of defence each may cooperate to protect an undefended
vertex. Fundamental structural properties were established by Chellali et
al.~\cite{chellali2016roman2}, while Italian domination on trees was studied
by Henning and Klostermeyer~\cite{henning2017italian}. Complexity and
algorithmic questions for Roman-\(\{2\}\) domination were subsequently
investigated by Chen and Lu~\cite{chen2019roman2} and by Padamutham and
Palagiri~\cite{padamutham2020complexity}. Despite this progress, its
complexity remained unsettled on several structured bipartite graph classes.

\paragraph*{\textbf{Double Roman domination.}}

Double Roman domination was introduced by Beeler, Haynes, and
Hedetniemi~\cite{beeler2016doubleRoman}. A double Roman dominating function
of \(G\) is a function \(f:V(G)\rightarrow\{0,1,2,3\}\) such that every
vertex assigned value \(0\) has either a neighbour assigned value \(3\), or
at least two neighbours assigned value \(2\), and every vertex assigned
value \(1\) has a neighbour assigned value \(2\) or \(3\). The minimum
weight of such a function is the double Roman domination number
\(\gamma_{dR}(G)\).

Double Roman domination models a stronger defence requirement than ordinary
Roman domination. A vertex assigned value \(3\) can provide the complete
defence required by an adjacent zero vertex, while two neighbours assigned
value \(2\) may jointly provide the same protection. Beeler et
al.~\cite{beeler2016doubleRoman} proved the useful fact that every graph has
a minimum double Roman dominating function in which no vertex is assigned
value \(1\). Therefore, when designing an optimization algorithm, it is
sufficient to consider the values \(0\), \(2\), and \(3\). Further
structural, complexity, and algorithmic results were obtained by Banerjee et al.~\cite{banerjee2020doubleRoman}, Yue et
al.~\cite{yue2018doubleRoman}, and Poureidi~\cite{poureidi2022doubleRoman}.

\paragraph*{\textbf{Perfect Roman domination.}}

A perfect Roman dominating function of a graph \(G\) is a function
\(f:V(G)\rightarrow\{0,1,2\}\) such that every vertex assigned value \(0\)
has exactly one neighbour assigned value \(2\). The minimum weight of such
a function is the perfect Roman domination number, denoted by
\(\gamma_R^p(G)\). Thus, perfect Roman domination strengthens ordinary Roman
domination by replacing the requirement of at least one value-\(2\)
neighbour with an exact-one requirement.

Perfect Roman domination was studied on trees by Henning et al.~\cite{henning2018perfectRomanTrees}. Banerjee et al.~\cite{banerjee2019perfectRoman} later developed a broader
algorithmic and complexity study of the problem and established results on
several restricted graph classes. The exactness requirement makes the
problem algorithmically different from ordinary Roman domination. Once a
zero vertex has acquired a value-\(2\) neighbour, every subsequent choice
must ensure that it does not acquire a second one. Consequently, a dynamic
program must retain information about both undominated zero vertices and
zero vertices whose protection has already been completed.

\paragraph*{\textbf{Unique response Roman domination.}}

The unique response form of Roman domination originates in the Roman
influence parameters introduced by Rubalcaba and
Slater~\cite{rubalcaba2007romanInfluence}. It was subsequently studied by
Targhi et al.~\cite{targhi2011uniqueResponse}. A
unique response Roman dominating function, abbreviated as a URRDF, is a
Roman dominating function \(f:V(G)\rightarrow\{0,1,2\}\) satisfying the
following additional conditions: every vertex assigned value \(0\) has
exactly one neighbour assigned value \(2\), and every vertex assigned value
\(1\) or \(2\) has no neighbour assigned value \(2\). The minimum weight of
a URRDF is the unique response Roman domination number \(u_R(G)\).

The first condition guarantees that every undefended vertex has a unique
possible responder. The second condition prevents a vertex assigned value
\(2\) from being adjacent to any positive vertex, including another
value-\(2\) vertex. Hence, the set of value-\(2\) vertices forms a highly
restricted packing. Relations between Roman domination and the unique response
Roman domination were studied by Targhi et
al.~\cite{targhi2011uniqueResponse}, while trees satisfying
\(\gamma_R(G)=u_R(G)\) were investigated by Jafari Rad and
Liu~\cite{rad2012uniqueTrees}. Zhao et al.~\cite{zhao2018uniqueTrees}
provided further results for trees and proved hardness on triangle-free
graphs.

A detailed algorithmic study of unique response Roman domination was
recently given by Banerjee et al.~\cite{banerjee2023unique}. They proved that the decision problem is
NP-complete on chordal graphs and on bipartite graphs, established strong
inapproximability results, and gave polynomial-time algorithms for
distance-hereditary graphs, interval graphs, and proper interval graphs. In
particular, they obtained a linear-time algorithm for distance-hereditary
graphs, a polynomial-time algorithm for interval graphs, and a linear-time
algorithm for proper interval graphs.

\subsection{Structured bipartite graph classes and motivation}

We study the four parameters on convex bipartite graphs and chordal
bipartite graphs. A bipartite graph \(G=(X\cup Y,E)\) is convex with respect
to \(X\) if there is a linear ordering
\(X=\{x_1,\ldots,x_m\}\) such that the neighbourhood of every vertex of
\(Y\) is an interval in this ordering. Thus, for every \(y\in Y\), there
exist indices \(l(y)\) and \(r(y)\) such that
\(N(y)=\{x_{l(y)},\ldots,x_{r(y)}\}\). The interval representation gives a
strong left-to-right ordering that is well-suited to dynamic programming.

A bipartite graph is chordal bipartite if it contains no induced cycle of
length at least \(6\). Convex bipartite graphs form a structured subclass
of chordal bipartite graphs, but a general chordal bipartite graph need not
admit a convex ordering. Consequently, the two classes provide a natural
setting for studying the effect of an interval ordering on the complexity of
Roman-type domination. Standard structural information about these and
related graph classes can be found in
\cite{brandstadt1999graphclasses}.

For ordinary Roman domination, polynomial-time algorithms are known on
several graph classes possessing suitable ordering or decomposition
properties
\cite{liedloff2008efficient,liu2013roman,padamutham2020algorithmic}.
More recently, Das et al.~\cite{das2026roman} studied Roman
domination specifically on subclasses of bipartite graphs. They gave an
\(O(n^3)\)-time algorithm for convex bipartite graphs and proved that the
problem is NP-complete on chordal bipartite graphs. These results reveal a
sharp algorithmic contrast between the interval structure of convex
bipartite graphs and the weaker induced-cycle restriction defining chordal
bipartite graphs. The present paper investigates whether the same contrast
persists for stronger Roman-type parameters.

The complexity table of Banerjee et
al.~\cite{banerjee2023unique} provides additional motivation for our study.
Their table records Roman-\(\{2\}\) domination, perfect Roman domination,
and unique response Roman domination as open on chordal bipartite graphs.
Double Roman domination was already known to be NP-complete on chordal
bipartite graphs~\cite{banerjee2020doubleRoman}. The same table records
Roman-\(\{2\}\), double Roman, perfect Roman, and unique response Roman
domination as open on bipartite permutation graphs. Since every bipartite
permutation graph is biconvex and hence convex bipartite, a polynomial-time
algorithm for convex bipartite graphs settles all four of these open entries
in the stronger convex-bipartite setting.

It is important to distinguish these statements carefully. Convex bipartite
graphs were not listed as a separate row in the table of Banerjee et
al.~\cite{banerjee2023unique}. However, their open cases for bipartite
permutation graphs are covered by our more general convex-bipartite
algorithms. On the hardness side, their table explicitly left the
Roman-\(\{2\}\), perfect Roman, and unique response Roman problems open on
chordal bipartite graphs, and these are precisely the three chordal
bipartite cases settled in this paper.

\subsection{Our Contributions}\label{subsec:our-contributions}

Our first contribution is a unified dynamic programming framework for
Roman-\(\{2\}\) domination, double Roman domination, perfect Roman
domination, and unique response Roman domination on convex bipartite graphs.
Let \(G=(X\cup Y,E)\) be convex with respect to the ordering
\(X=\{x_1,\ldots,x_m\}\). For every \(y\in Y\), write
\(N(y)=\{x_{l(y)},\ldots,x_{r(y)}\}\), and define
\(Y_i=\{y\in Y:r(y)=i\}\). The algorithm processes the graph in the order
\(x_1,Y_1,x_2,Y_2,\ldots,x_m,Y_m\).

When a vertex \(y\in Y_i\) is processed, all its neighbours have already
been processed, and its neighbourhood
\(N(y)=\{x_{l(y)},\ldots,x_i\}\) is a suffix of the currently processed
part of \(X\). This suffix property allows the unfinished requirements of
the processed vertices to be compressed into a constant number of boundary
indices. The interpretation of the indices depends on the variant. For
Roman-\(\{2\}\) and double Roman domination, the state records unresolved
deficits of processed zero vertices. For a perfect Roman and a unique response
Roman domination, it records active vertices that still need a value-\(2\)
neighbour together with the vertices that must not be reached by another
value-\(2\) vertex.

Although the four state spaces are different, they all contain five
coordinates, each having \(O(n)\) possible values. Every vertex is processed
once, and every transition takes constant time. Consequently, each of the
four parameters can be computed in \(O(n^6)\) time and \(O(n^5)\) space on
an \(n\)-vertex convex bipartite graph. Since biconvex bipartite graphs and
bipartite permutation graphs are subclasses of convex bipartite graphs,
these algorithms also establish polynomial-time solvability on both of
these subclasses.

Our second contribution is a collection of hardness results for chordal
bipartite graphs. We prove that Roman-\(\{2\}\) domination is NP-complete on
chordal bipartite graphs by a reduction from \textsc{Dominating Set} on
chordal bipartite graphs, which is known to be NP-complete
\cite{muller1987np}. We prove that unique response Roman domination is
NP-complete on chordal bipartite graphs by a reduction based on
\textsc{Efficient Domination}; the required hardness result for efficient
domination follows from the work of Lu and Tang~\cite{luTang2002}. Finally,
we establish the NP-completeness of perfect Roman domination on chordal
bipartite graphs using a separate edge gadget that enforces the independence
condition required by the reduction.

These results settle the chordal-bipartite entries that were explicitly
listed as open for Roman-\(\{2\}\), perfect Roman, and unique response Roman
domination in the complexity table of Banerjee et
al.~\cite{banerjee2023unique}. Together with the existing hardness of double
Roman domination on chordal bipartite graphs completes the complexity
classification of the four variants considered here on chordal bipartite
graphs. At the same time, the convex-bipartite algorithms settle the open
bipartite-permutation cases recorded in the same table.

Therefore, our results establish a clear tractability–hardness contrast.
The interval ordering of convex bipartite graphs permits the unfinished
local conditions to be represented by compact boundary information, leading
to polynomial-time algorithms. By contrast, excluding induced cycles of
length at least \(6\) does not provide enough structure to make the same
problems tractable on chordal bipartite graphs.

The remainder of the paper is organized as follows.
Section~\ref{sec:preliminaries} introduces the necessary definitions and
notation. Section~\ref{sec:convex-roman-type} presents the unified dynamic
programming framework for the four Roman-type domination problems on convex
bipartite graphs. Section~\ref{sec:chordal-bipartite-hardness} establishes
the NP-completeness results for Roman-\(\{2\}\), perfect Roman, and
unique response Roman domination on chordal bipartite graphs.
Section~\ref{sec:conclusion} concludes the paper and discusses directions
for further research.

\section{Preliminaries}\label{sec:preliminaries}

In this section, we recall the graph-theoretic terminology and the Roman-type domination parameters used throughout the paper. All graphs considered in this paper are finite, simple, and undirected. For a graph \(G\), we denote its vertex set and edge set by \(V(G)\) and \(E(G)\), respectively. For a vertex \(v\in V(G)\), the open neighborhood of \(v\) is \(N_G(v)=\{u\in V(G):uv\in E(G)\}\), and the closed neighborhood is \(N_G[v]=N_G(v)\cup\{v\}\). When the graph is clear from the context, we simply write \(N(v)\) and \(N[v]\).

A set \(D\subseteq V(G)\) is a dominating set of \(G\) if every vertex in \(V(G)\setminus D\) has a neighbour in \(D\). The minimum cardinality of a dominating set of \(G\) is the domination number of \(G\), denoted by \(\gamma(G)\). A set \(D\subseteq V(G)\) is an efficient dominating set if every vertex \(v\in V(G)\) satisfies \(|N_G[v]\cap D|=1\).

A Roman dominating function of \(G\) is a function \(f:V(G)\rightarrow\{0,1,2\}\) such that every vertex \(v\) with \(f(v)=0\) has at least one neighbor \(u\) with \(f(u)=2\). The weight of \(f\) is \(w(f)=\sum_{v\in V(G)}f(v)\). The Roman domination number of \(G\), denoted by \(\gamma_R(G)\), is the minimum weight of a Roman dominating function of \(G\).

A Roman-\(\{2\}\) dominating function of \(G\) is a function \(f:V(G)\rightarrow\{0,1,2\}\) such that every vertex \(v\) with \(f(v)=0\) satisfies \(\sum_{u\in N(v)}f(u)\geq 2\). The Roman-\(\{2\}\) domination number of \(G\), denoted by \(\gamma_{\{R2\}}(G)\), is the minimum weight of such a function.

A double Roman dominating function of \(G\) is a function \(f:V(G)\rightarrow\{0,1,2,3\}\) such that every vertex assigned value \(0\) has either a neighbour assigned value \(3\), or at least two neighbours assigned value \(2\), and every vertex assigned value \(1\) has a neighbour assigned value \(2\) or \(3\). The double Roman domination number of \(G\), denoted by \(\gamma_{dR}(G)\), is the minimum weight of a double Roman dominating function. We use the standard fact that every graph has a minimum double Roman dominating function using no vertex of value \(1\).

A perfect Roman dominating function of \(G\) is a function \(f:V(G)\rightarrow\{0,1,2\}\) such that every vertex assigned value \(0\) has exactly one neighbour assigned value \(2\). The perfect Roman domination number of \(G\), denoted by \(\gamma_R^p(G)\), is the minimum weight of a perfect Roman dominating function.

A unique response Roman dominating function of \(G\) is a Roman dominating function \(f:V(G)\rightarrow\{0,1,2\}\) with the following additional restrictions: every vertex assigned value \(0\) has exactly one neighbour assigned value \(2\), and every vertex assigned value \(1\) or \(2\) has no neighbor assigned value \(2\). The unique response Roman domination number of \(G\) is denoted by \(u_R(G)\).

A bipartite graph \(G=(X\cup Y,E)\) is convex with respect to \(X\) if there exists a linear ordering \(X=\{x_1,x_2,\ldots,x_m\}\) such that, for every \(y\in Y\), the neighborhood of \(y\) is an interval in this ordering. Thus, for every \(y\in Y\), there exist integers \(l(y)\) and \(r(y)\), with \(1\leq l(y)\leq r(y)\leq m\), such that \(N(y)=\{x_{l(y)},x_{l(y)+1},\ldots,x_{r(y)}\}\). Throughout the algorithmic section, we assume that such a convex ordering and the interval endpoints \(l(y)\) and \(r(y)\) are given. For each \(i\in\{1,\ldots,m\}\), we write \(Y_i=\{y\in Y:r(y)=i\}\).

A bipartite graph is chordal bipartite if it has no induced cycle of length at least six. Equivalently, every induced cycle in a chordal bipartite graph has length four. Chordal bipartite graphs strictly generalize several structured bipartite graph classes, but they do not necessarily admit the interval ordering used in convex bipartite graphs. This distinction is important in this paper: the interval structure of convex bipartite graphs leads to polynomial-time algorithms, while the broader chordal bipartite class remains hard for several Roman-type domination problems.

For each Roman-type parameter considered in this paper, the corresponding decision problem asks whether the graph admits a feasible function of weight at most a given integer \(K\). Since a proposed function can be checked in polynomial time, all these decision problems belong to \(\mathsf{NP}\).

\section{A Unified Dynamic Programming Framework for Roman-Type Domination in Convex Bipartite Graphs}\label{sec:convex-roman-type}

In this section, we show that the left-to-right dynamic programming method used for ordinary Roman domination on convex bipartite graphs can be extended to four related Roman-type domination parameters. The four parameters considered here are Roman-\(\{2\}\) domination, double Roman domination, perfect Roman domination, and unique response Roman domination. Although their local feasibility conditions are different, all four algorithms use the same structural principle: the graph is processed from left to right along a convex ordering of one bipartition class, and the unfinished requirements of already processed vertices are represented by a constant number of boundary indices.

\subsection{Convex Ordering and Processing Order}

Let \(G=(X\cup Y,E)\) be a bipartite graph, where \(X=\{x_1,x_2,\ldots,x_m\}\) is given in a fixed linear order. We assume that \(G\) is convex with respect to \(X\). Thus, for every vertex \(y\in Y\), the neighbourhood of \(y\) is an interval in the ordering of \(X\), and we write \(N(y)=\{x_{l(y)},x_{l(y)+1},\ldots,x_{r(y)}\}\). For each \(i\in\{1,\ldots,m\}\), let \(Y_i=\{y\in Y:r(y)=i\}\).

The vertices are processed in the order \(x_1,Y_1,x_2,Y_2,\ldots,x_m,Y_m\). Inside each set \(Y_i\), the vertices may be processed in an arbitrary but fixed order. When a vertex \(y\in Y_i\) is processed, all its neighbours have already been processed, and its neighbourhood is \(N(y)=\{x_{l(y)},\ldots,x_i\}\). Therefore, the neighbourhood of \(y\) is a suffix of the currently processed part \(\{x_1,\ldots,x_i\}\) of \(X\). This suffix property is the central structural reason why the unfinished requirements of processed vertices of \(X\) can be compressed into a small number of indices.

For a fixed algorithm, the table \(D\) stores the minimum weight of a feasible partial assignment for every state that can occur after the part of the graph processed so far. We use \(p_\infty=m+1\) as a special state value, meaning that the corresponding vertex does not exist. The symbol \(+\infty\), on the other hand, denotes an infeasible table entry. A state \(\sigma\) is called \emph{finite} if \(D(\sigma)<+\infty\). Whenever a new vertex is processed, a temporary table \(D'\) is initialized with the value \(+\infty\) in every entry. All feasible transitions from the current table \(D\) are written into \(D'\), and then \(D'\) replaces \(D\). The instruction ``update \(D'(\sigma')\) with \(w\)'' means that \(D'(\sigma')\) is replaced by \(\min\{D'(\sigma'),w\}\).

For every variant, all already processed vertices of \(Y\) must completely satisfy their domination condition. Indeed, when \(y\in Y_i\) is processed, all vertices of \(N(y)\) have already been assigned values, and no future vertex can change the contribution received by \(y\). In contrast, a processed vertex of \(X\) assigned value \(0\) may still receive a contribution from a future vertex of \(Y\). Consequently, the states store unfinished information only for processed vertices of \(X\).

\subsection{A Common Deficit Representation}

Roman-\(\{2\}\) domination and double Roman domination use the same form of unfinished information. In both cases, a processed vertex \(x_t\in X\) assigned value \(0\) may still need either two units or one unit of contribution from future vertices of \(Y\). We say that \(x_t\) has \emph{deficit \(2\)} if it has not yet received any useful contribution, and that it has \emph{deficit \(1\)} if it has received one useful unit but is not yet safe.

The state stores two indices \(a\) and \(b\). The value \(a\) is the smallest index of a processed vertex of \(X\) having deficit \(2\); if no such vertex exists, then \(a=p_\infty\). If \(a\neq p_\infty\), then \(b\) is the smallest index of a deficit-\(1\) vertex lying strictly before \(a\); if no such vertex exists, then \(b=p_\infty\). Deficit-\(1\) vertices at or after \(a\) do not need to be recorded separately. If \(a=p_\infty\), then \(b\) is simply the smallest index of any deficit-\(1\) vertex, or \(p_\infty\) if no such vertex exists.

The reason this compressed pair is sufficient is the suffix property. Suppose that \(x_a\) is the earliest deficit-\(2\) vertex. Every future vertex of \(Y\) that reaches \(x_a\) also reaches every later processed vertex of \(X\). Hence, as long as \(x_a\) remains unfinished, the detailed status of later deficient vertices does not affect any future feasibility decision. The only deficit-\(1\) vertex that may behave differently is one lying before \(a\), and it is enough to store the earliest such vertex in \(b\). When no deficit-\(2\) vertex remains, the earliest deficit-\(1\) vertex alone determines whether a future suffix update clears all remaining deficits.

We use two constant-time update operations on the pair \((a,b)\). The operation \(\operatorname{Reduce}(a,b,l)\) is applied when a processed vertex \(y\in Y_i\) gives one useful unit to every unfinished vertex in the suffix \(\{x_l,\ldots,x_i\}\). The operation \(\operatorname{Clear}(a,b,l)\) is applied when \(y\) gives enough contribution to make every unfinished vertex in that suffix safe.

For \(\operatorname{Reduce}(a,b,l)\), the following cases apply.
\begin{enumerate}
	\item If \(a=p_\infty\), then only deficit-\(1\) vertices may remain. If \(b<l\), the earliest such vertex is outside the suffix and the new pair is \((p_\infty,b)\). If \(b\geq l\), every deficit-\(1\) vertex lies in the suffix and becomes safe, so the new pair is \((p_\infty,p_\infty)\).
	\item If \(a\neq p_\infty\) and \(a<l\), the earliest deficit-\(2\) vertex is outside the suffix. It remains the controlling unfinished vertex, and the new pair is \((a,b)\).
	\item If \(a\neq p_\infty\) and \(a\geq l\), every deficit-\(2\) vertex lies in the suffix and is reduced to deficit \(1\). If \(b<l\), the previously stored deficit-\(1\) vertex remains the earliest one, so the new pair is \((p_\infty,b)\). Otherwise, the vertex \(x_a\) becomes the earliest remaining deficit-\(1\) vertex, and the new pair is \((p_\infty,a)\).
\end{enumerate}

For \(\operatorname{Clear}(a,b,l)\), the following cases apply.
\begin{enumerate}
	\item If \(a=p_\infty\), then the new pair is \((p_\infty,b)\) when \(b<l\), and \((p_\infty,p_\infty)\) otherwise.
	\item If \(a\neq p_\infty\) and \(a<l\), the earliest deficit-\(2\) vertex is not reached, and the new pair remains \((a,b)\).
	\item If \(a\neq p_\infty\) and \(a\geq l\), every deficit-\(2\) vertex is cleared. If \(b<l\), the stored deficit-\(1\) vertex remains unfinished and the new pair is \((p_\infty,b)\). Otherwise, every unfinished vertex is cleared, and the new pair is \((p_\infty,p_\infty)\).
\end{enumerate}

Thus, \(\operatorname{Reduce}\) lowers by one the unfinished requirements reached by the current suffix, whereas \(\operatorname{Clear}\) removes all unfinished requirements reached by that suffix.

\subsection{Roman-\texorpdfstring{\(\{2\}\)}{2} Domination}

A Roman-\(\{2\}\) dominating function of \(G\) is a function \(f:V(G)\rightarrow\{0,1,2\}\) such that every vertex \(v\) with \(f(v)=0\) satisfies \(\sum_{u\in N(v)}f(u)\geq 2\). The weight of \(f\) is \(w(f)=\sum_{v\in V(G)}f(v)\), and the Roman-\(\{2\}\) domination number is denoted by \(\gamma_{\{R2\}}(G)\).

A zero vertex can be protected either by one neighbour of value \(2\), or by at least two neighbours of value \(1\). Therefore, when a vertex \(y\in Y_i\) is processed, the algorithm must know whether its interval contains a value-\(2\) vertex of \(X\), or at least two value-\(1\) vertices of \(X\).

\paragraph*{State.}
A state is a tuple \((r,u,v,a,b)\). The coordinate \(r\) is the largest index of a processed vertex of \(X\) assigned value \(2\); if no such vertex exists, then \(r=0\). The coordinates \(u\) and \(v\) are respectively the largest and second largest indices of processed vertices of \(X\) assigned value \(1\); a missing index is represented by \(0\). The pair \((a,b)\) stores the deficit information defined above.

For a finite state \((r,u,v,a,b)\), the table value \(D(r,u,v,a,b)\) is the minimum weight of a feasible partial Roman-\(\{2\}\) assignment realizing these five boundary values. Feasibility means that every processed vertex of \(Y\) assigned value \(0\) already receives a total neighbour value at least \(2\), while every unfinished processed vertex of \(X\) is represented by \((a,b)\).

The initial state is \((0,0,0,p_\infty,p_\infty)\) with table value \(0\). A final state is accepting exactly when \(a=b=p_\infty\), because then every processed zero vertex of \(X\) has received total contribution at least \(2\).

\paragraph*{Processing a vertex of \texorpdfstring{\(X\)}.}
Suppose that \(x_i\) is processed from a finite state \((r,u,v,a,b)\).
\begin{enumerate}
	\item If \(f(x_i)=0\), then \(x_i\) initially has deficit \(2\), since no previously processed vertex of \(Y\) is adjacent to \(x_i\). If \(a=p_\infty\), set \(a_0=i\); otherwise set \(a_0=a\). The new state is \((r,u,v,a_0,b)\), with no increase in weight.
	\item If \(f(x_i)=1\), then \(x_i\) is safe and becomes the newest value-\(1\) vertex of \(X\). The new state is \((r,i,u,a,b)\), and the weight increases by \(1\).
	\item If \(f(x_i)=2\), then \(x_i\) is safe and becomes the newest value-\(2\) vertex of \(X\). The new state is \((i,u,v,a,b)\), and the weight increases by \(2\).
\end{enumerate}

\paragraph*{Processing a vertex of \texorpdfstring{\(Y\)}.}
Let \(y\in Y_i\), and write \(l=l(y)\). Then \(N(y)=\{x_l,\ldots,x_i\}\).
\begin{enumerate}
	\item If \(f(y)=0\), then \(y\) must receive a total neighbour value at least \(2\). The interval contains a value-\(2\) vertex exactly when \(r\geq l\), and it contains at least two value-\(1\) vertices exactly when \(v\geq l\). Hence, the transition is feasible exactly when \(r\geq l\) or \(v\geq l\). In this case, the state and weight do not change.
	\item If \(f(y)=1\), then \(y\) contributes one unit to every unfinished zero vertex of \(X\) in its interval. The deficit pair becomes \(\operatorname{Reduce}(a,b,l)\), and the weight increases by \(1\).
	\item If \(f(y)=2\), then every unfinished zero vertex of \(X\) in its interval receives enough contribution to become safe. The deficit pair becomes \(\operatorname{Clear}(a,b,l)\), and the weight increases by \(2\).
\end{enumerate}

\subsection{Double Roman Domination}

A double Roman dominating function of \(G\) is a function \(f:V(G)\rightarrow\{0,1,2,3\}\) such that every vertex assigned value \(0\) has either a neighbour assigned value \(3\), or at least two neighbours assigned value \(2\), and every vertex assigned value \(1\) has a neighbour assigned value \(2\) or \(3\). The double Roman domination number is denoted by \(\gamma_{dR}(G)\).

We use the standard observation that every graph has a minimum double Roman dominating function in which no vertex is assigned value \(1\). Indeed, if \(f(v)=1\) and \(v\) has a neighbour \(u\) of value \(3\), then changing \(v\) from \(1\) to \(0\) decreases the weight. If \(f(v)=1\) and \(v\) has a neighbour \(u\) of value \(2\), then changing \(v\) from \(1\) to \(0\) and \(u\) from \(2\) to \(3\) preserves the weight and decreases the number of vertices assigned value \(1\). Thus, among the minimum solutions, one may assume that only the values \(0\), \(2\), and \(3\) are used.

\paragraph*{State.}
A state is a tuple \((r_3,r_2,s_2,a,b)\). The coordinate \(r_3\) is the largest index of a processed vertex of \(X\) assigned value \(3\). The coordinates \(r_2\) and \(s_2\) are respectively the largest and second largest indices of processed vertices of \(X\) assigned value \(2\). Missing indices are represented by \(0\). The pair \((a,b)\) stores the unfinished double Roman requirements: deficit \(2\) means that a zero vertex has not yet received any value-\(2\) neighbour from \(Y\), while deficit \(1\) means that it has received one such neighbour but still needs either another value-\(2\) neighbour or a value-\(3\) neighbour.

The initial state is \((0,0,0,p_\infty,p_\infty)\) with value \(0\), and a final state is accepting exactly when \(a=b=p_\infty\).

\paragraph*{Processing a vertex of \texorpdfstring{\(X\)}.}
Suppose that \(x_i\) is processed from a finite state \((r_3,r_2,s_2,a,b)\).
\begin{enumerate}
	\item If \(f(x_i)=0\), then \(x_i\) starts with deficit \(2\). Thus, if \(a=p_\infty\), set \(a_0=i\); otherwise set \(a_0=a\). The new state is \((r_3,r_2,s_2,a_0,b)\).
	\item If \(f(x_i)=2\), then \(x_i\) becomes the newest value-\(2\) vertex of \(X\). The new state is \((r_3,i,r_2,a,b)\), and the weight increases by \(2\).
	\item If \(f(x_i)=3\), then \(x_i\) becomes the newest value-\(3\) vertex of \(X\). The new state is \((i,r_2,s_2,a,b)\), and the weight increases by \(3\).
\end{enumerate}

\paragraph*{Processing a vertex of \texorpdfstring{\(Y\)}.}
Let \(y\in Y_i\), and let \(l=l(y)\).
\begin{enumerate}
	\item If \(f(y)=0\), then \(y\) must have either a value-\(3\) neighbour or two value-\(2\) neighbours. The first condition is equivalent to \(r_3\geq l\), and the second is equivalent to \(s_2\geq l\). Hence, this transition is feasible exactly when \(r_3\geq l\) or \(s_2\geq l\).
	\item If \(f(y)=2\), then \(y\) gives one useful unit to every unfinished zero vertex of \(X\) in its interval. The deficit pair becomes \(\operatorname{Reduce}(a,b,l)\), and the weight increases by \(2\).
	\item If \(f(y)=3\), then every unfinished zero vertex in its interval becomes safe. The deficit pair becomes \(\operatorname{Clear}(a,b,l)\), and the weight increases by \(3\).
\end{enumerate}

\subsection{Perfect Roman Domination}

A perfect Roman dominating function of \(G\) is a function \(f:V(G)\rightarrow\{0,1,2\}\) such that every vertex assigned value \(0\) has exactly one neighbour assigned value \(2\). The perfect Roman domination number is denoted by \(\gamma_R^p(G)\).

For a processed vertex \(x_t\in X\) assigned value \(0\), we use two terms. The vertex \(x_t\) is \emph{active} if it has not yet received any neighbour in \(Y\) assigned value \(2\). It is \emph{closed} if it has already received exactly one such neighbour. A closed vertex must never be reached again by a future vertex of \(Y\) assigned value \(2\).

\paragraph*{State.}
A state is a tuple \((r_1,r_2,q,p,t)\). The coordinates \(r_1\) and \(r_2\) are respectively the largest and second largest indices of processed vertices of \(X\) assigned value \(2\); missing indices are represented by \(0\). The coordinate \(q\) is the largest index of a closed vertex of \(X\), and \(q=0\) if no closed vertex exists. The coordinates \(p\) and \(t\) are respectively the first and last active vertices of \(X\). If no active vertex exists, then \(p=t=p_\infty\).

The pair \((p,t)\) is sufficient because the suffix structure permits only three useful possibilities when a future value-\(2\) vertex of \(Y\) is processed: it reaches all active vertices, it reaches none of them, or it reaches only a proper suffix. The third possibility can never be extended to a perfect Roman dominating function. Indeed, after a proper suffix is closed, any future interval that reaches the earliest still-active vertex also reaches a later vertex that is already closed, thereby giving that closed vertex a second value-\(2\) neighbour.

The initial state is \((0,0,0,p_\infty,p_\infty)\) with value \(0\). A final state is accepting exactly when \(p=t=p_\infty\).

\paragraph*{Processing a vertex of \texorpdfstring{\(X\)}.}
Suppose that \(x_i\) is processed from a finite state \((r_1,r_2,q,p,t)\).
\begin{enumerate}
	\item If \(f(x_i)=0\), then \(x_i\) becomes active. If \(p=p_\infty\), the new active boundaries are \(p=t=i\). Otherwise, \(p\) remains unchanged and \(t\) becomes \(i\).
	\item If \(f(x_i)=1\), then \(x_i\) is safe by itself. The state does not change, and the weight increases by \(1\).
	\item If \(f(x_i)=2\), then \(x_i\) becomes the newest value-\(2\) vertex of \(X\). The new state is \((i,r_1,q,p,t)\), and the weight increases by \(2\).
\end{enumerate}

\paragraph*{Processing a vertex of \texorpdfstring{\(Y\)}.}
Let \(y\in Y_i\), and let \(l=l(y)\).
\begin{enumerate}
	\item If \(f(y)=0\), then \(y\) must have exactly one value-\(2\) neighbour in \(N(y)=\{x_l,\ldots,x_i\}\). Since \(r_1\) and \(r_2\) are the rightmost and second rightmost value-\(2\) vertices of \(X\), this condition is equivalent to \(r_1\geq l\) and \(r_2<l\).
	\item If \(f(y)=1\), then \(y\) is safe and does not affect any active or closed vertex of \(X\). The state is unchanged, and the weight increases by \(1\).
	\item Suppose that \(f(y)=2\). First, \(y\) must avoid every closed vertex of \(X\). Since \(q\) is the largest closed index and \(N(y)\) is a suffix, this condition is exactly \(l>q\). If \(l\leq q\), the transition is infeasible. Assume that \(l>q\). If \(p=p_\infty\), then no active vertex exists, and the state is unchanged. If \(l\leq p\), then \(y\) reaches the earliest active vertex, and hence every active vertex; all active vertices become closed, the new closed boundary is \(t\), and the new state is \((r_1,r_2,t,p_\infty,p_\infty)\). If \(l>t\), then \(y\) reaches no active vertex and the state is unchanged. Finally, if \(p<l\leq t\), then \(y\) closes only a proper suffix of the active vertices, and the transition is discarded because it cannot be extended to a perfect Roman dominating function. In every feasible case, the weight increases by \(2\).
\end{enumerate}

\subsection{unique response Roman Domination}

A unique response Roman dominating function of \(G\) is a Roman dominating function \(f:V(G)\rightarrow\{0,1,2\}\) satisfying the additional conditions that every vertex assigned value \(0\) has exactly one neighbour assigned value \(2\), and every vertex assigned value \(1\) or \(2\) has no neighbour assigned value \(2\). The unique response Roman domination number is denoted by \(u_R(G)\).

\paragraph*{State.}
The state is again a tuple \((r_1,r_2,q,p,t)\). The coordinates \(r_1\) and \(r_2\) are the largest and second largest indices of processed vertices of \(X\) assigned value \(2\). The coordinates \(p\) and \(t\) are the first and last active vertices of \(X\), where active means assigned value \(0\) but not yet dominated by a value-\(2\) vertex of \(Y\).

The meaning of \(q\) is stronger than in perfect Roman domination. A processed vertex of \(X\) is called \emph{forbidden} if it is assigned value \(1\) or \(2\), or if it is assigned value \(0\) and is already closed. A future vertex of \(Y\) assigned value \(2\) must avoid every forbidden vertex. The coordinate \(q\) is the largest index of a forbidden vertex, and \(q=0\) if no forbidden vertex exists.

The initial state is \((0,0,0,p_\infty,p_\infty)\) with value \(0\). A final state is accepting exactly when \(p=t=p_\infty\).

\paragraph*{Processing a vertex of \texorpdfstring{\(X\)}.}
Suppose that \(x_i\) is processed from a finite state \((r_1,r_2,q,p,t)\).
\begin{enumerate}
	\item If \(f(x_i)=0\), then \(x_i\) becomes active. If \(p=p_\infty\), the new active boundaries are \(p=t=i\); otherwise, \(p\) remains unchanged and \(t\) becomes \(i\).
	\item If \(f(x_i)=1\), then \(x_i\) becomes forbidden. This transition is feasible only when \(p=p_\infty\). Otherwise, any future value-\(2\) vertex of \(Y\) that reaches the earliest active vertex also reaches \(x_i\), contradicting the unique response condition. When the transition is feasible, the new state is \((r_1,r_2,i,p_\infty,p_\infty)\), and the weight increases by \(1\).
	\item If \(f(x_i)=2\), then \(x_i\) becomes forbidden and also becomes the newest value-\(2\) vertex of \(X\). This transition is likewise feasible only when \(p=p_\infty\). The new state is \((i,r_1,i,p_\infty,p_\infty)\), and the weight increases by \(2\).
\end{enumerate}

\paragraph*{Processing a vertex of \texorpdfstring{\(Y\)}.}
Let \(y\in Y_i\), and let \(l=l(y)\).
\begin{enumerate}
	\item If \(f(y)=0\), then \(y\) must have exactly one value-\(2\) neighbour in \(N(y)\). This is equivalent to \(r_1\geq l\) and \(r_2<l\).
	\item If \(f(y)=1\), then \(y\) must have no value-\(2\) neighbour. Since \(r_1\) is the rightmost value-\(2\) vertex of \(X\), this condition is equivalent to \(r_1<l\). When it holds, the state is unchanged and the weight increases by \(1\).
	\item Suppose that \(f(y)=2\). The vertex \(y\) must avoid all forbidden vertices of \(X\), which is equivalent to \(l>q\). If \(l\leq q\), the transition is infeasible. Assume that \(l>q\). If \(p=p_\infty\), then no active vertex exists and the state is unchanged. If \(l\leq p\), then \(y\) reaches and closes every active vertex; these vertices become forbidden, the new forbidden boundary is \(t\), and the new state is \((r_1,r_2,t,p_\infty,p_\infty)\). If \(l>t\), then \(y\) reaches no active vertex and the state remains unchanged. If \(p<l\leq t\), then \(y\) closes only a proper suffix of the active vertices, and the transition is discarded because any future value-\(2\) vertex needed to close the earliest active vertex would also reach a later forbidden vertex. In each feasible case, the weight increases by \(2\).
\end{enumerate}

\subsection{Pseudocode}

We divide the unified framework into three parts. Algorithm~\ref{alg:roman-type-convex} is the main routine, Algorithm~\ref{alg:process-x-roman-type} processes a vertex of \(X\), and Algorithm~\ref{alg:process-y-roman-type} processes a vertex of \(Y\). The parameter \(\Pi\) specifies the problem: \(\mathrm{R2}\), \(\mathrm{DR}\), \(\mathrm{PR}\), or \(\mathrm{UR}\).

A state is called finite if its table entry is smaller than \(+\infty\). In either transition subroutine, the temporary table \(D'\) is initialized with \(+\infty\) in every entry. Whenever several transitions produce the same state, the update convention keeps only the minimum resulting weight. The accepting states have deficit coordinates \((p_\infty,p_\infty)\) for \(\mathrm{R2}\) and \(\mathrm{DR}\), and active-boundary coordinates \((p_\infty,p_\infty)\) for \(\mathrm{PR}\) and \(\mathrm{UR}\).

\begin{algorithm}[H]
	\DontPrintSemicolon
	\caption{\textsc{Roman-Type-Convex-Bipartite}}
	\label{alg:roman-type-convex}
	\KwIn{A convex bipartite graph \(G=(X\cup Y,E)\), where \(X=\{x_1,\ldots,x_m\}\), the interval representation \(N(y)=\{x_{l(y)},\ldots,x_{r(y)}\}\) for every \(y\in Y\), and a type \(\Pi\in\{\mathrm{R2},\mathrm{DR},\mathrm{PR},\mathrm{UR}\}\).}
	\KwOut{The corresponding Roman-type domination number.}
	
	Let \(p_\infty=m+1\)\;
	\For{\(i=1\) \KwTo \(m\)}{Set \(Y_i=\emptyset\)\;}
	\ForEach{\(y\in Y\)}{Append \(y\) to \(Y_{r(y)}\)\;}
	Set all entries of \(D\) to \(+\infty\)\;
	Set \(D(0,0,0,p_\infty,p_\infty)=0\)\;
	
	\For{\(i=1\) \KwTo \(m\)}{
		Set \(D=\textsc{Process-Roman-Type-X}(D,i,\Pi)\)\;
		\ForEach{\(y\in Y_i\), in an arbitrary fixed order}{
			Set \(D=\textsc{Process-Roman-Type-Y}(D,y,\Pi)\)\;
		}
	}
	
	\If{\(\Pi=\mathrm{R2}\)}{\Return{\(\min\{D(r,u,v,p_\infty,p_\infty):0\leq r,u,v\leq m\}\)}\;}
	\If{\(\Pi=\mathrm{DR}\)}{\Return{\(\min\{D(r_3,r_2,s_2,p_\infty,p_\infty):0\leq r_3,r_2,s_2\leq m\}\)}\;}
	\If{\(\Pi=\mathrm{PR}\)}{\Return{\(\min\{D(r_1,r_2,q,p_\infty,p_\infty):0\leq r_1,r_2,q\leq m\}\)}\;}
	\If{\(\Pi=\mathrm{UR}\)}{\Return{\(\min\{D(r_1,r_2,q,p_\infty,p_\infty):0\leq r_1,r_2,q\leq m\}\)}\;}
\end{algorithm}

\begin{algorithm}[H]
	\DontPrintSemicolon
	\caption{\textsc{Process-Roman-Type-X}}
	\label{alg:process-x-roman-type}
	\KwIn{The current table \(D\), an index \(i\), and the type \(\Pi\).}
	\KwOut{The updated table after processing \(x_i\).}
	
	Set all entries of \(D'\) to \(+\infty\)\;
	\ForEach{finite state \(\sigma\) of \(D\)}{
		\If{\(\Pi=\mathrm{R2}\)}{
			Write \(\sigma=(r,u,v,a,b)\)\;
			\eIf{\(a=p_\infty\)}{Set \(a_0=i\)\;}{Set \(a_0=a\)\;}
			Update \(D'(r,u,v,a_0,b)\) with \(D(r,u,v,a,b)\)\tcp*{\(f(x_i)=0\)}
			Update \(D'(r,i,u,a,b)\) with \(D(r,u,v,a,b)+1\)\tcp*{\(f(x_i)=1\)}
			Update \(D'(i,u,v,a,b)\) with \(D(r,u,v,a,b)+2\)\tcp*{\(f(x_i)=2\)}
		}
		\If{\(\Pi=\mathrm{DR}\)}{
			Write \(\sigma=(r_3,r_2,s_2,a,b)\)\;
			\eIf{\(a=p_\infty\)}{Set \(a_0=i\)\;}{Set \(a_0=a\)\;}
			Update \(D'(r_3,r_2,s_2,a_0,b)\) with \(D(r_3,r_2,s_2,a,b)\)\tcp*{\(f(x_i)=0\)}
			Update \(D'(r_3,i,r_2,a,b)\) with \(D(r_3,r_2,s_2,a,b)+2\)\tcp*{\(f(x_i)=2\)}
			Update \(D'(i,r_2,s_2,a,b)\) with \(D(r_3,r_2,s_2,a,b)+3\)\tcp*{\(f(x_i)=3\)}
		}
		\If{\(\Pi=\mathrm{PR}\)}{
			Write \(\sigma=(r_1,r_2,q,p,t)\)\;
			\eIf{\(p=p_\infty\)}{
				Update \(D'(r_1,r_2,q,i,i)\) with \(D(r_1,r_2,q,p,t)\)\tcp*{\(f(x_i)=0\)}
			}{
				Update \(D'(r_1,r_2,q,p,i)\) with \(D(r_1,r_2,q,p,t)\)\tcp*{\(f(x_i)=0\)}
			}
			Update \(D'(r_1,r_2,q,p,t)\) with \(D(r_1,r_2,q,p,t)+1\)\tcp*{\(f(x_i)=1\)}
			Update \(D'(i,r_1,q,p,t)\) with \(D(r_1,r_2,q,p,t)+2\)\tcp*{\(f(x_i)=2\)}
		}
		\If{\(\Pi=\mathrm{UR}\)}{
			Write \(\sigma=(r_1,r_2,q,p,t)\)\;
			\eIf{\(p=p_\infty\)}{
				Update \(D'(r_1,r_2,q,i,i)\) with \(D(r_1,r_2,q,p,t)\)\tcp*{\(f(x_i)=0\)}
				Update \(D'(r_1,r_2,i,p_\infty,p_\infty)\) with \(D(r_1,r_2,q,p,t)+1\)\tcp*{\(f(x_i)=1\)}
				Update \(D'(i,r_1,i,p_\infty,p_\infty)\) with \(D(r_1,r_2,q,p,t)+2\)\tcp*{\(f(x_i)=2\)}
			}{
				Update \(D'(r_1,r_2,q,p,i)\) with \(D(r_1,r_2,q,p,t)\)\tcp*{\(f(x_i)=0\)}
			}
		}
	}
	\Return{\(D'\)}\;
\end{algorithm}

\begin{algorithm}[H]
	\DontPrintSemicolon
	\caption{\textsc{Process-Roman-Type-Y}}
	\label{alg:process-y-roman-type}
	\KwIn{The current table \(D\), a vertex \(y\in Y\), and the type \(\Pi\).}
	\KwOut{The updated table after processing \(y\).}
	
	Set \(l=l(y)\)\;
	Set all entries of \(D'\) to \(+\infty\)\;
	\ForEach{finite state \(\sigma\) of \(D\)}{
		\If{\(\Pi=\mathrm{R2}\)}{
			Write \(\sigma=(r,u,v,a,b)\)\;
			\If{\(r\geq l\) or \(v\geq l\)}{Update \(D'(r,u,v,a,b)\) with \(D(r,u,v,a,b)\)\tcp*{\(f(y)=0\)}}
			Compute \((a_1,b_1)=\operatorname{Reduce}(a,b,l)\)\;
			Update \(D'(r,u,v,a_1,b_1)\) with \(D(r,u,v,a,b)+1\)\tcp*{\(f(y)=1\)}
			Compute \((a_2,b_2)=\operatorname{Clear}(a,b,l)\)\;
			Update \(D'(r,u,v,a_2,b_2)\) with \(D(r,u,v,a,b)+2\)\tcp*{\(f(y)=2\)}
		}
		\If{\(\Pi=\mathrm{DR}\)}{
			Write \(\sigma=(r_3,r_2,s_2,a,b)\)\;
			\If{\(r_3\geq l\) or \(s_2\geq l\)}{Update \(D'(r_3,r_2,s_2,a,b)\) with \(D(r_3,r_2,s_2,a,b)\)\tcp*{\(f(y)=0\)}}
			Compute \((a_2,b_2)=\operatorname{Reduce}(a,b,l)\)\;
			Update \(D'(r_3,r_2,s_2,a_2,b_2)\) with \(D(r_3,r_2,s_2,a,b)+2\)\tcp*{\(f(y)=2\)}
			Compute \((a_3,b_3)=\operatorname{Clear}(a,b,l)\)\;
			Update \(D'(r_3,r_2,s_2,a_3,b_3)\) with \(D(r_3,r_2,s_2,a,b)+3\)\tcp*{\(f(y)=3\)}
		}
		\If{\(\Pi=\mathrm{PR}\)}{
			Write \(\sigma=(r_1,r_2,q,p,t)\)\;
			\If{\(r_1\geq l\) and \(r_2<l\)}{Update \(D'(r_1,r_2,q,p,t)\) with \(D(r_1,r_2,q,p,t)\)\tcp*{\(f(y)=0\)}}
			Update \(D'(r_1,r_2,q,p,t)\) with \(D(r_1,r_2,q,p,t)+1\)\tcp*{\(f(y)=1\)}
			\If{\(l>q\)}{
				\If{\(p=p_\infty\)}{Update \(D'(r_1,r_2,q,p,t)\) with \(D(r_1,r_2,q,p,t)+2\)\tcp*{\(f(y)=2\)}}
				\ElseIf{\(l\leq p\)}{Update \(D'(r_1,r_2,t,p_\infty,p_\infty)\) with \(D(r_1,r_2,q,p,t)+2\)\tcp*{\(f(y)=2\)}}
				\ElseIf{\(l>t\)}{Update \(D'(r_1,r_2,q,p,t)\) with \(D(r_1,r_2,q,p,t)+2\)\tcp*{\(f(y)=2\)}}
			}
		}
		\If{\(\Pi=\mathrm{UR}\)}{
			Write \(\sigma=(r_1,r_2,q,p,t)\)\;
			\If{\(r_1\geq l\) and \(r_2<l\)}{Update \(D'(r_1,r_2,q,p,t)\) with \(D(r_1,r_2,q,p,t)\)\tcp*{\(f(y)=0\)}}
			\If{\(r_1<l\)}{Update \(D'(r_1,r_2,q,p,t)\) with \(D(r_1,r_2,q,p,t)+1\)\tcp*{\(f(y)=1\)}}
			\If{\(l>q\)}{
				\If{\(p=p_\infty\)}{Update \(D'(r_1,r_2,q,p,t)\) with \(D(r_1,r_2,q,p,t)+2\)\tcp*{\(f(y)=2\)}}
				\ElseIf{\(l\leq p\)}{Update \(D'(r_1,r_2,t,p_\infty,p_\infty)\) with \(D(r_1,r_2,q,p,t)+2\)\tcp*{\(f(y)=2\)}}
				\ElseIf{\(l>t\)}{Update \(D'(r_1,r_2,q,p,t)\) with \(D(r_1,r_2,q,p,t)+2\)\tcp*{\(f(y)=2\)}}
			}
		}
	}
	\Return{\(D'\)}\;
\end{algorithm}

\subsection{Correctness Proof}

We prove the four algorithms simultaneously through a common invariant. For each \(i\in\{1,\ldots,m\}\), let \(X_i=\{x_1,\ldots,x_i\}\) and \(Y_{\leq i}=\{y\in Y:r(y)\leq i\}\).

\begin{lemma}\label{lem:roman-type-dp-invariant}
	After every processing step, for each state \(\sigma\), the entry \(D(\sigma)\) is the minimum weight of a feasible partial assignment realizing \(\sigma\). If no feasible partial assignment realizes \(\sigma\), then \(D(\sigma)=+\infty\).
\end{lemma}

\begin{proof}
	The claim is proved by induction on the processing order. Before any vertex is processed, the only feasible assignment is the empty assignment, represented by \((0,0,0,p_\infty,p_\infty)\) with value \(0\). Thus, the invariant holds initially.
	
	Assume that the invariant holds before the next vertex is processed. If the next vertex is \(x_i\), then no already processed vertex of \(Y\) is adjacent to \(x_i\). Consequently, assigning value \(0\) creates exactly the new unfinished requirement described in the state definition of the chosen variant. Assigning a positive value makes \(x_i\) safe and updates the rightmost positive indices exactly as stated. In the unique response case, a positive value is allowed only when no active vertex exists; otherwise, every future suffix that reaches the earlier active vertex also reaches \(x_i\), making a valid completion impossible. Hence, the transitions for \(x_i\) are sound and exhaustive.
	
	Now suppose that the next vertex is \(y\in Y_i\). Since \(N(y)=\{x_{l(y)},\ldots,x_i\}\), all neighbours of \(y\) have already been processed. Therefore, the feasibility of assigning value \(0\) or, in the unique response case, value \(1\), can be decided exactly from the stored rightmost positive indices. For Roman-\(\{2\}\) domination, the test \(r\geq l(y)\) or \(v\geq l(y)\) is equivalent to the existence of one value-\(2\) neighbour or two value-\(1\) neighbours. For double Roman domination, the test \(r_3\geq l(y)\) or \(s_2\geq l(y)\) is equivalent to the existence of one value-\(3\) neighbour or two value-\(2\) neighbours. For perfect and unique response Roman domination, the condition \(r_1\geq l(y)\) and \(r_2<l(y)\) is equivalent to having exactly one value-\(2\) neighbour. In the unique response case, \(r_1<l(y)\) is equivalent to having no value-\(2\) neighbour.
	
	For Roman-\(\{2\}\) and double Roman domination, the effect of a positive value on the unfinished vertices of \(X\) is exactly described by \(\operatorname{Reduce}\) or \(\operatorname{Clear}\). The suffix property guarantees that the compressed pair \((a,b)\) retains all information relevant to future transitions.
	
	For perfect Roman domination, a value-\(2\) vertex of \(Y\) must avoid closed vertices, and the condition \(l(y)>q\) is equivalent to this requirement. It either reaches all active vertices, reaches none, or reaches only a proper suffix. The first two possibilities are represented by the stated transitions. The third possibility cannot be extended, because any future suffix reaching the earliest remaining active vertex also reaches a later vertex that has already been closed. The unique response case is identical after replacing closed vertices by forbidden vertices. Thus, every listed transition is sound, and every feasible extension of the current partial assignment appears among the listed transitions.
	
	Finally, when several partial assignments produce the same new state, only the one of minimum weight is retained. This is safe because all future transitions depend only on the state and not on the internal details of the partial assignment. Hence, the invariant is preserved.
\end{proof}

\begin{theorem}\label{thm:roman-type-convex-correct}
	Algorithm~\ref{alg:roman-type-convex} computes \(\gamma_{\{R2\}}(G)\), \(\gamma_{dR}(G)\), \(\gamma_R^p(G)\), or \(u_R(G)\), according to the chosen value of \(\Pi\).
\end{theorem}

\begin{proof}
	After all vertices have been processed, every vertex of \(Y\) already satisfies its local domination requirement, because this requirement was checked when that vertex was processed. Thus, only unfinished vertices of \(X\) can prevent a final state from representing a complete solution.
	
	For Roman-\(\{2\}\) and double Roman domination, no unfinished zero vertex remains exactly when \(a=b=p_\infty\). For perfect Roman and unique response Roman domination, no active zero vertex remains exactly when \(p=t=p_\infty\). Hence, the final states accepted by the main algorithm are precisely the complete feasible functions of the corresponding type. By Lemma~\ref{lem:roman-type-dp-invariant}, each accepted table entry stores the minimum weight among all functions realizing that state. Taking the minimum over all accepted states, therefore, gives the required domination number.
\end{proof}

\subsection{Running Time Analysis}

\begin{theorem}\label{thm:roman-type-convex-running}
	Let \(G=(X\cup Y,E)\) be a convex bipartite graph with a given convex ordering of \(X\), and let \(n=|V(G)|\). Then Roman-\(\{2\}\) domination, double Roman domination, perfect Roman domination, and unique response Roman domination can each be computed in \(O(n^6)\) time and \(O(n^5)\) space.
\end{theorem}

\begin{proof}
	Let \(m=|X|\) and \(q=|Y|\), so \(n=m+q\). The sets \(Y_i=\{y\in Y:r(y)=i\}\) can be constructed in \(O(m+q)\) time when the interval endpoints are given. If the endpoints must be extracted from adjacency lists, the preprocessing takes \(O(|E|)\) time, which is at most \(O(n^2)\) and does not dominate the dynamic programming time.
	
	Each of the four states has five coordinates, and every coordinate has at most \(m+1\) possible values. Therefore, each table contains \(O(m^5)\) states. Every processed vertex gives rise to at most three transitions from each finite state. Each transition performs only a constant number of comparisons, state-coordinate changes, and table updates. In particular, \(\operatorname{Reduce}\) and \(\operatorname{Clear}\) take constant time.
	
	The total time for the \(m\) vertices of \(X\) is \(O(m\cdot m^5)=O(m^6)\). The total time for the \(q\) vertices of \(Y\) is \(O(qm^5)\). Hence the total running time is \(O(m^6+qm^5)=O((m+q)m^5)=O(n^6)\). Only the current table \(D\) and the temporary table \(D'\) are stored, so the space requirement is \(O(m^5)=O(n^5)\).
\end{proof}

\subsection{Example}

We illustrate the four variants on the same convex bipartite graph. Let \(X=\{x_1,x_2,x_3,x_4\}\) with ordering \(x_1,x_2,x_3,x_4\), and let \(Y=\{y_1,y_2,y_3\}\). Define \(N(y_1)=\{x_1,x_2,x_3\}\), \(N(y_2)=\{x_2,x_3,x_4\}\), and \(N(y_3)=\{x_3,x_4\}\). Then \(r(y_1)=3\), \(r(y_2)=4\), and \(r(y_3)=4\). Hence \(Y_1=Y_2=\emptyset\), \(Y_3=\{y_1\}\), and \(Y_4=\{y_2,y_3\}\), and the processing order is \(x_1,x_2,x_3,y_1,x_4,y_2,y_3\).

\begin{figure}[htbp!]
	\centering
	\begin{tikzpicture}[
		xnode/.style={circle,draw,minimum size=7mm,inner sep=0pt},
		ynode/.style={circle,draw,minimum size=7mm,inner sep=0pt},
		edge/.style={thick},
		every node/.style={font=\small}
		]
		\node[xnode] (x1) at (0,3) {\(x_1\)};
		\node[xnode] (x2) at (0,2) {\(x_2\)};
		\node[xnode] (x3) at (0,1) {\(x_3\)};
		\node[xnode] (x4) at (0,0) {\(x_4\)};
		\node[ynode] (y1) at (3,2.5) {\(y_1\)};
		\node[ynode] (y2) at (3,1.3) {\(y_2\)};
		\node[ynode] (y3) at (3,0.4) {\(y_3\)};
		\node at (0,3.7) {\(X\)};
		\node at (3,3.7) {\(Y\)};
		\draw[edge] (x1)--(y1);
		\draw[edge] (x2)--(y1);
		\draw[edge] (x3)--(y1);
		\draw[edge] (x2)--(y2);
		\draw[edge] (x3)--(y2);
		\draw[edge] (x4)--(y2);
		\draw[edge] (x3)--(y3);
		\draw[edge] (x4)--(y3);
	\end{tikzpicture}
	\caption{The convex bipartite graph used in the example. The neighbourhoods of \(y_1\), \(y_2\), and \(y_3\) are intervals in the ordering \(x_1,x_2,x_3,x_4\).}
	\label{fig:roman-type-convex-example}
\end{figure}
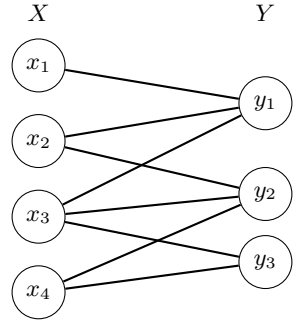

\begin{table}[htbp]
	\centering
	\small
	\renewcommand{\arraystretch}{1.25}
	\begin{tabularx}{\linewidth}{>{\raggedright\arraybackslash}p{0.26\linewidth}>{\raggedright\arraybackslash}X>{\raggedright\arraybackslash}p{0.20\linewidth}}
		\hline
		Parameter & One optimal assignment & Optimum value \\
		\hline
		Roman-\(\{2\}\) domination & \(f(y_1)=2\), \(f(x_4)=2\); all other vertices receive \(0\). & \(\gamma_{\{R2\}}(G)=4\) \\
		\hline
		Double Roman domination & \(f(y_1)=3\), \(f(x_4)=3\); all other vertices receive \(0\). & \(\gamma_{dR}(G)=6\) \\
		\hline
		Perfect Roman domination & \(f(y_1)=2\), \(f(x_4)=2\); all other vertices receive \(0\). & \(\gamma_R^p(G)=4\) \\
		\hline
		unique response Roman domination & \(f(y_1)=2\), \(f(x_4)=2\); all other vertices receive \(0\). & \(u_R(G)=4\) \\
		\hline
	\end{tabularx}
	\caption{Optimal assignments for the four Roman-type domination parameters on the graph in Figure~\ref{fig:roman-type-convex-example}.}
	\label{tab:roman-type-convex-example}
\end{table}

For Roman-\(\{2\}\) domination, assign value \(2\) to \(y_1\) and \(x_4\), and value \(0\) to every other vertex. Each of \(x_1,x_2,x_3\) receives total neighbour value \(2\) from \(y_1\), while \(y_2\) and \(y_3\) receive total neighbour value \(2\) from \(x_4\). Hence, the assignment is feasible and has weight \(4\). A direct examination of assignments of weight at most \(3\) shows that none can simultaneously provide total neighbour value at least \(2\) to the vertices at both ends of the convex ordering. Therefore \(\gamma_{\{R2\}}(G)=4\).

For double Roman domination, assign value \(3\) to \(y_1\) and \(x_4\), and value \(0\) to every other vertex. The vertices \(x_1,x_2,x_3\) are defended by \(y_1\), while \(y_2\) and \(y_3\) are defended by \(x_4\). Thus, the weight is \(6\). Since \(x_1\) has the unique neighbour \(y_1\), either \(x_1\) is positive or \(y_1\) has value \(3\). The analogous requirements near \(x_4\), together with the conditions for \(y_2\) and \(y_3\), force total weight at least \(6\). Hence \(\gamma_{dR}(G)=6\).

For perfect Roman domination, the assignment with \(f(y_1)=2\), \(f(x_4)=2\), and all other values equal to \(0\) is feasible. Each of \(x_1,x_2,x_3\) has exactly one value-\(2\) neighbour, namely \(y_1\), while each of \(y_2,y_3\) has exactly one value-\(2\) neighbour, namely \(x_4\). Since every perfect Roman dominating function is also a Roman dominating function and the ordinary Roman domination number of this graph is \(4\), the optimum is \(\gamma_R^p(G)=4\).

Finally, the same assignment is a unique response Roman dominating function. Every zero vertex has exactly one value-\(2\) neighbour, and the two value-\(2\) vertices \(y_1\) and \(x_4\) are nonadjacent. Therefore, no positive vertex has a value-\(2\) neighbour. Again, every unique response Roman dominating function is a Roman dominating function, so no solution of weight below \(4\) exists. Hence \(u_R(G)=4\).


\section{Hardness of Roman-Type Domination Problems on Chordal Bipartite Graphs}\label{sec:chordal-bipartite-hardness}

In this section, we prove hardness results for several Roman-type domination problems on chordal bipartite graphs. We first prove NP-completeness for Roman-\(\{2\}\) domination by reducing from \textsc{Dominating Set} on chordal bipartite graphs. We then prove NP-completeness for unique response Roman domination by reducing from \textsc{Efficient Domination} on chordal bipartite graphs. Finally, we prove NP-completeness for perfect Roman domination by using a separate edge gadget that charges one extra unit whenever both endpoints of an original edge are selected. The Roman-\(\{2\}\) reduction uses one pendant vertex per original vertex, while the unique response Roman reduction uses two pendant vertices per original vertex.

\begin{figure}[htbp]
	\centering
	\resizebox{0.85\linewidth}{!}{%
		\begin{tikzpicture}[
			vtx/.style={circle,draw,minimum size=7mm,inner sep=0pt},
			oldv/.style={circle,draw,minimum size=4.5mm,inner sep=0pt},
			pend/.style={circle,draw,minimum size=7mm,inner sep=0pt},
			oldedge/.style={thick,dashed},
			newedge/.style={thick},
			every node/.style={font=\small}
			]
			
			\node at (0,2.0) {\(A\)-side};
			\node at (3,2.0) {\(B\)-side};
			\node[vtx] (v1) at (0,0) {\(v\)};
			\node[oldv] (u11) at (3,1.0) {};
			\node[oldv] (u12) at (3,0.35) {};
			\node[oldv] (u13) at (3,-0.30) {};
			\node at (3,-0.85) {\(N_G(v)\)};
			\node[pend] (p1) at (3,-1.65) {\(p_v\)};
			\draw[oldedge] (v1)--(u11);
			\draw[oldedge] (v1)--(u12);
			\draw[oldedge] (v1)--(u13);
			\draw[newedge] (v1)--(p1);
			\node[align=center] at (1.5,-2.55) {(a) Roman-\(\{2\}\) reduction\\one pendant vertex};
			
			\draw[dotted,thick] (4.25,2.25)--(4.25,-2.85);
			
			\node at (5.5,2.0) {\(A\)-side};
			\node at (8.5,2.0) {\(B\)-side};
			\node[vtx] (v2) at (5.5,0) {\(v\)};
			\node[oldv] (u21) at (8.5,1.0) {};
			\node[oldv] (u22) at (8.5,0.35) {};
			\node[oldv] (u23) at (8.5,-0.30) {};
			\node at (8.5,-0.85) {\(N_G(v)\)};
			\node[pend] (p2) at (8.5,-1.50) {\(p_v\)};
			\node[pend] (q2) at (8.5,-2.25) {\(q_v\)};
			\draw[oldedge] (v2)--(u21);
			\draw[oldedge] (v2)--(u22);
			\draw[oldedge] (v2)--(u23);
			\draw[newedge] (v2)--(p2);
			\draw[newedge] (v2)--(q2);
			\node[align=center] at (7.0,-3.15) {(b) unique response Roman reduction\\two pendant vertices};
	\end{tikzpicture}}
	\caption{Local gadgets used in the chordal-bipartite reductions. Dashed edges represent the original edges from \(v\) to its neighbours in \(G\), and solid edges represent the newly added pendant edges. The figure is drawn for the case \(v\in A\); when \(v\in B\), the construction is symmetric, and the new pendant vertices are placed on the \(A\)-side.}
	\label{fig:chordal-bipartite-roman-type-reductions}
\end{figure}
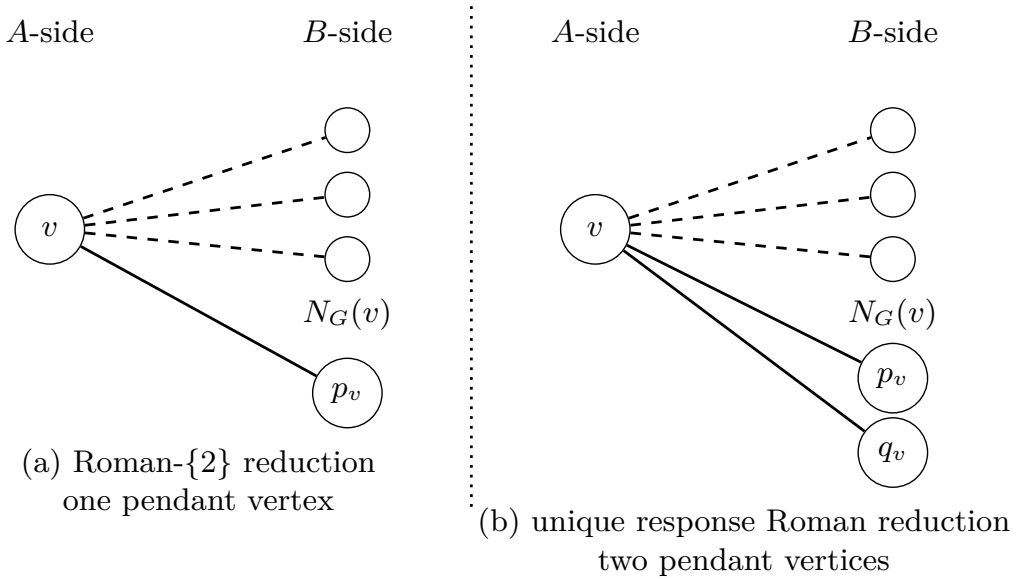

\subsection{Roman-\texorpdfstring{\(\{2\}\)}{2} Domination}

\begin{theorem}
	\textsc{Roman-\(\{2\}\) Domination} is NP-complete on chordal bipartite graphs.
\end{theorem}

\begin{proof}
	The problem belongs to \(\mathsf{NP}\). We reduce from \textsc{Dominating Set} on chordal bipartite graphs, which is NP-complete~\cite{muller1987np}. Let \((G,k)\) be an instance of \textsc{Dominating Set}, where \(G=(A\cup B,E)\) is chordal bipartite, and let \(n=|V(G)|\). We construct a graph \(H\) from \(G\) as follows. For every vertex \(v\in V(G)\), add one new pendant vertex \(p_v\) adjacent only to \(v\); see Figure~\ref{fig:chordal-bipartite-roman-type-reductions}(a). If \(v\in A\), then put \(p_v\) in the \(B\)-side, and if \(v\in B\), then put \(p_v\) in the \(A\)-side. Thus, \(H\) is bipartite. Moreover, adding pendant vertices cannot create an induced cycle of length at least six, because no pendant vertex lies on a cycle. Hence, \(H\) is chordal bipartite. We set \(K=n+k\).
	
	Suppose first that \(G\) has a dominating set \(D\) of size at most \(k\). Define \(f:V(H)\rightarrow\{0,1,2\}\) as follows. For every \(v\in D\), set \(f(v)=2\) and \(f(p_v)=0\). For every \(v\notin D\), set \(f(v)=0\) and \(f(p_v)=1\). The weight is \(2|D|+(n-|D|)=n+|D|\leq n+k\). If \(p_v=0\), then \(v\in D\), and \(p_v\) has the neighbour \(v\) of value \(2\). If \(v\notin D\), then \(p_v\) has value \(1\), and since \(D\) dominates \(G\), the vertex \(v\) has a neighbor \(u\in D\). Thus \(v\) receives value \(1\) from \(p_v\) and value \(2\) from \(u\), and hence receives total neighbour value at least \(3\). Therefore, every zero vertex receives a total neighbour value of at least \(2\), and \(f\) is feasible.
	
	Conversely, suppose that \(H\) has a Roman-\(\{2\}\) dominating function \(f\) of weight at most \(n+k\). If \(f(p_v)=2\), then replace the values on \(\{v,p_v\}\) by \(f(v)=2\) and \(f(p_v)=0\). This does not increase the weight and preserves feasibility. Hence, we may assume that every pendant vertex has value \(0\) or \(1\).
	
	Now each pair \(\{v,p_v\}\) has total weight at least \(1\). Define \(D=\{v\in V(G):f(v)>0\}\). Every vertex in \(D\) contributes at least one unit above the baseline value \(1\), and therefore \(|D|\leq w(f)-n\leq k\). Let \(v\in V(G)\setminus D\). Then \(f(v)=0\). Since \(p_v\) is a pendant vertex and \(f(v)\neq 2\), we must have \(f(p_v)=1\). The vertex \(v\) must receive a total neighbour value of at least \(2\). The pendant \(p_v\) contributes exactly \(1\), so \(v\) must receive positive contribution from some original neighbor \(u\). Hence \(u\in D\). Therefore, \(D\) is a dominating set of \(G\) of size at most \(k\). This proves NP-hardness.
\end{proof}

\subsection{unique response Roman Domination}

\begin{theorem}
	\textsc{unique response Roman Domination} is NP-complete on chordal bipartite graphs.
\end{theorem}

\begin{proof}
	The problem belongs to \(\mathsf{NP}\). We reduce from \textsc{Efficient Domination} on chordal bipartite graphs, which is NP-complete~\cite{luTang2002}. Recall that \(D\subseteq V(G)\) is an efficient dominating set if every vertex \(v\in V(G)\) satisfies \(|N_G[v]\cap D|=1\).
	
	Let \(G=(A\cup B,E)\) be a chordal bipartite graph with \(n=|V(G)|\). We construct \(H\) by adding, for every \(v\in V(G)\), two pendant vertices \(p_v\) and \(q_v\), both adjacent only to \(v\); see Figure~\ref{fig:chordal-bipartite-roman-type-reductions}(b). If \(v\in A\), then \(p_v\) and \(q_v\) are placed in the \(B\)-side, and if \(v\in B\), then they are placed in the \(A\)-side. The graph \(H\) is chordal bipartite. We set \(K=2n\).
	
	If \(G\) has an efficient dominating set \(D\), define \(f(v)=2\) and \(f(p_v)=f(q_v)=0\) for every \(v\in D\), and define \(f(v)=0\) and \(f(p_v)=f(q_v)=1\) for every \(v\notin D\). Since \(D\) is efficient, every vertex outside \(D\) has exactly one neighbour in \(D\), and no two vertices of \(D\) are adjacent. Hence, every zero vertex has exactly one neighbour of value \(2\), and every vertex of value \(1\) or \(2\) has no neighbour of value \(2\). Thus \(f\) is a unique response Roman dominating function of weight \(2n\).
	
	Conversely, suppose that \(H\) has a unique response Roman dominating function \(f\) of weight at most \(2n\). For each \(v\in V(G)\), the triple \(\{v,p_v,q_v\}\) has weight at least \(2\). Hence, every triple has weight exactly \(2\). Therefore, either \(f(v)=2\) and \(f(p_v)=f(q_v)=0\), or \(f(v)=0\) and \(f(p_v)=f(q_v)=1\). Let \(D=\{v\in V(G):f(v)=2\}\). Since two vertices assigned value \(2\) cannot be adjacent in a unique response Roman dominating function, \(D\) is independent. If \(v\notin D\), then \(f(v)=0\), and its two pendant neighbors have value \(1\). Hence, \(v\) must have exactly one original neighbour in \(D\). Thus, every vertex outside \(D\) has exactly one neighbour in \(D\), and every vertex in \(D\) is dominated exactly once by itself. Therefore, \(D\) is an efficient dominating set of \(G\). This proves NP-hardness.
\end{proof}

\subsection{Perfect Roman Domination}

The remaining variant is Perfect Roman domination. The two-pendant construction used for unique response Roman domination is not sufficient here, because in a perfect Roman dominating function, two adjacent vertices are allowed to both receive value \(2\). Therefore, we add an edge gadget that charges one extra unit whenever both endpoints of an original edge are selected. This restores the independence condition needed to simulate efficient domination.

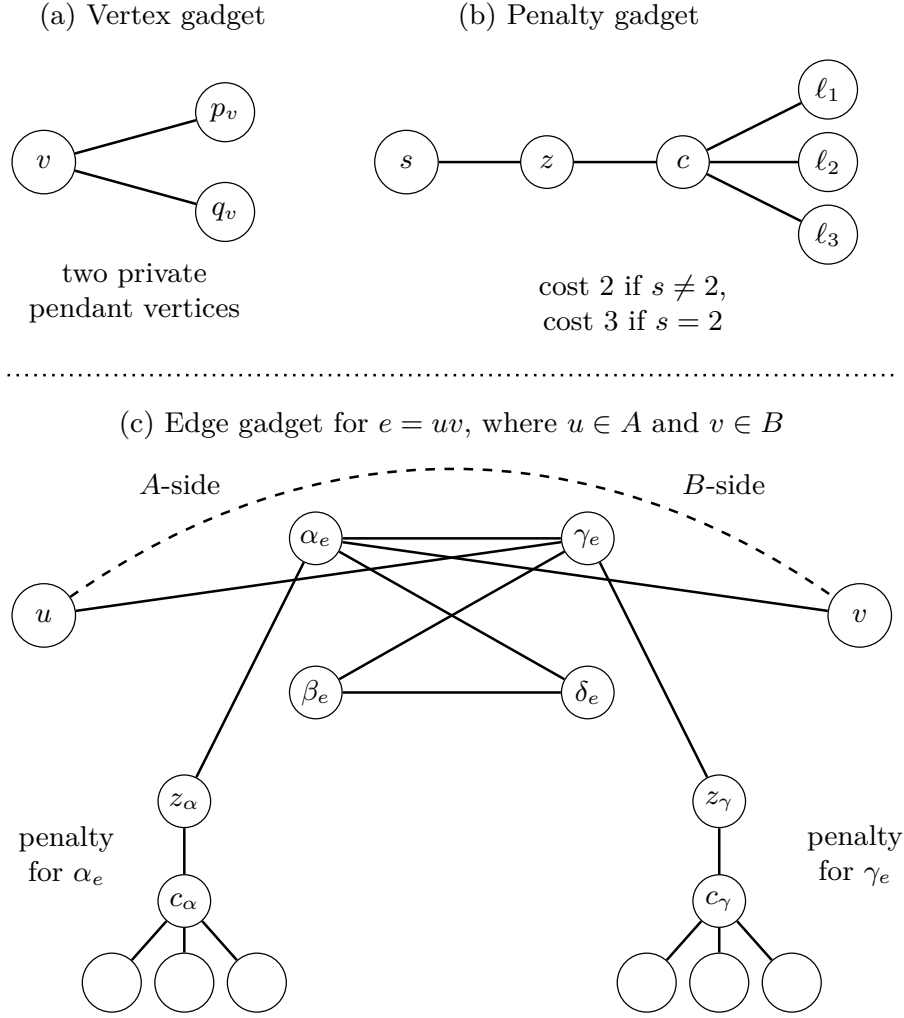
\begin{figure}[t]
	\centering
	
	\resizebox{0.75\linewidth}{!}{%
		\begin{tikzpicture}[vtx/.style={circle,draw,minimum size=7mm,inner sep=0pt},
			smallv/.style={circle,draw,minimum size=5.8mm,inner sep=0pt},
			pend/.style={circle,draw,minimum size=6.5mm,inner sep=0pt},
			oldedge/.style={thick,dashed},
			newedge/.style={thick},
			every node/.style={font=\small}
			]
			
			\node at (1.2,1.6) {(a) Vertex gadget};
			\node[vtx] (v) at (0,0) {\(v\)};
			\node[pend] (pv) at (2,0.55) {\(p_v\)};
			\node[pend] (qv) at (2,-0.55) {\(q_v\)};
			\draw[newedge] (v)--(pv);
			\draw[newedge] (v)--(qv);
			\node[align=center] at (1.0,-1.45) {two private\\pendant vertices};
			
			\node at (5.9,1.6) {(b) Penalty gadget};
			\node[vtx] (s) at (4,0) {\(s\)};
			\node[smallv] (z) at (5.55,0) {\(z\)};
			\node[smallv] (c) at (7.05,0) {\(c\)};
			\node[pend] (l1) at (8.65,0.8) {\(\ell_1\)};
			\node[pend] (l2) at (8.65,0) {\(\ell_2\)};
			\node[pend] (l3) at (8.65,-0.8) {\(\ell_3\)};
			\draw[newedge] (s)--(z);
			\draw[newedge] (z)--(c);
			\draw[newedge] (c)--(l1);
			\draw[newedge] (c)--(l2);
			\draw[newedge] (c)--(l3);
			\node[align=center] at (6.5,-1.55) {cost \(2\) if \(s\neq 2\),\\cost \(3\) if \(s=2\)};
			
			\draw[dotted,thick] (-0.4,-2.35)--(9.4,-2.35);
			
			\node at (4.5,-2.9) {(c) Edge gadget for \(e=uv\), where \(u\in A\) and \(v\in B\)};
			
			\node at (1.5,-3.55) {\(A\)-side};
			\node at (7.5,-3.55) {\(B\)-side};
			
			\node[vtx] (u) at (0,-5.0) {\(u\)};
			\node[vtx] (v2) at (9.0,-5.0) {\(v\)};
			
			\node[smallv] (alpha) at (3.0,-4.15) {\(\alpha_e\)};
			\node[smallv] (beta) at (3.0,-5.85) {\(\beta_e\)};
			\node[smallv] (gamma) at (6.0,-4.15) {\(\gamma_e\)};
			\node[smallv] (delta) at (6.0,-5.85) {\(\delta_e\)};
			
			\draw[oldedge] (u) to[out=35,in=145] (v2);
			
			\draw[newedge] (u)--(gamma);
			\draw[newedge] (alpha)--(v2);
			\draw[newedge] (alpha)--(gamma);
			\draw[newedge] (alpha)--(delta);
			\draw[newedge] (beta)--(gamma);
			\draw[newedge] (beta)--(delta);
			
			\node[smallv] (za) at (1.55,-7.05) {\(z_\alpha\)};
			\node[smallv] (ca) at (1.55,-8.15) {\(c_\alpha\)};
			\node[pend] (la1) at (0.75,-9.05) {};
			\node[pend] (la2) at (1.55,-9.05) {};
			\node[pend] (la3) at (2.35,-9.05) {};
			
			\draw[newedge] (alpha)--(za);
			\draw[newedge] (za)--(ca);
			\draw[newedge] (ca)--(la1);
			\draw[newedge] (ca)--(la2);
			\draw[newedge] (ca)--(la3);
			\node[align=center] at (0.25,-7.65) {penalty\\for \(\alpha_e\)};
			
			\node[smallv] (zg) at (7.45,-7.05) {\(z_\gamma\)};
			\node[smallv] (cg) at (7.45,-8.15) {\(c_\gamma\)};
			\node[pend] (lg1) at (6.65,-9.05) {};
			\node[pend] (lg2) at (7.45,-9.05) {};
			\node[pend] (lg3) at (8.25,-9.05) {};
			
			\draw[newedge] (gamma)--(zg);
			\draw[newedge] (zg)--(cg);
			\draw[newedge] (cg)--(lg1);
			\draw[newedge] (cg)--(lg2);
			\draw[newedge] (cg)--(lg3);
			\node[align=center] at (8.95,-7.65) {penalty\\for \(\gamma_e\)};
			
	\end{tikzpicture}}
	\caption{Gadgets used in the reduction for Perfect Roman domination. The vertex gadget is attached to every original vertex. The penalty gadget is attached to every incidence \((v,e)\), and also to the internal vertices \(\alpha_e\) and \(\gamma_e\) of each edge gadget. The dashed edge \(uv\) is the original edge of \(G\); all other edges shown in the edge gadget are newly added.}
	\label{fig:perfect-roman-chordal-bipartite-reduction}
\end{figure}

\begin{theorem}\label{thm:perfect-roman-chordal-bipartite}
	\textsc{Perfect Roman Domination} is NP-complete on chordal bipartite graphs.
\end{theorem}

\begin{proof}
	The problem belongs to \(\mathsf{NP}\). We reduce from \textsc{Efficient Domination} on chordal bipartite graphs, which is NP-complete~\cite{luTang2002}. Let \(G=(A\cup B,E)\) be a chordal bipartite graph, let \(n=|V(G)|\), and let \(m=|E(G)|\). We construct a graph \(H\) as follows.
	
	For every vertex \(v\in V(G)\), add two private pendant vertices \(p_v\) and \(q_v\), both adjacent only to \(v\). For every incidence of a vertex \(v\) with an edge \(e\in E(G)\), attach a penalty gadget to \(v\). This gadget consists of two vertices \(z_{v,e}\) and \(c_{v,e}\), together with three private leaves adjacent to \(c_{v,e}\), and with edges \(vz_{v,e}\), \(z_{v,e}c_{v,e}\), and the three leaf edges incident with \(c_{v,e}\). Finally, for every edge \(e=uv\in E(G)\), where \(u\in A\) and \(v\in B\), add two vertices \(\alpha_e,\beta_e\) to the \(A\)-side and two vertices \(\gamma_e,\delta_e\) to the \(B\)-side, and add the edges \(u\gamma_e\), \(\alpha_ev\), \(\alpha_e\gamma_e\), \(\alpha_e\delta_e\), \(\beta_e\gamma_e\), and \(\beta_e\delta_e\). We also attach one penalty gadget to \(\alpha_e\) and one penalty gadget to \(\gamma_e\). The construction is shown in Figure~\ref{fig:perfect-roman-chordal-bipartite-reduction}.
	
	The graph \(H\) is bipartite by construction. Pendant vertices and penalty gadgets are trees attached to existing vertices, so they cannot create an induced cycle. For each edge gadget, any long cycle using internal vertices is chorded either by the original edge \(uv\) or by the internal edge \(\alpha_e\gamma_e\). Since \(G\) is chordal bipartite, no induced cycle of length at least six is created. Hence \(H\) is chordal bipartite. We set \(K=2n+11m\).
	
	Suppose first that \(G\) has an efficient dominating set \(D\). For every \(v\in D\), set \(f(v)=2\) and \(f(p_v)=f(q_v)=0\). For every \(v\notin D\), set \(f(v)=0\) and \(f(p_v)=f(q_v)=1\). In every incidence penalty gadget attached to \(v\), set \(f(c_{v,e})=2\), set the three private leaves of \(c_{v,e}\) to \(0\), and set \(f(z_{v,e})=1\) if \(v\in D\), while \(f(z_{v,e})=0\) if \(v\notin D\).
	
	For an edge \(e=uv\), since \(D\) is independent, the case \(u,v\in D\) does not occur. If \(u\notin D\) and \(v\notin D\), set \(f(\alpha_e)=0\), \(f(\beta_e)=0\), \(f(\gamma_e)=1\), and \(f(\delta_e)=2\). If \(u\in D\) and \(v\notin D\), set \(f(\alpha_e)=0\), \(f(\beta_e)=0\), \(f(\gamma_e)=0\), and \(f(\delta_e)=2\). If \(u\notin D\) and \(v\in D\), set \(f(\alpha_e)=0\), \(f(\beta_e)=2\), \(f(\gamma_e)=0\), and \(f(\delta_e)=0\). The penalty gadgets attached to \(\alpha_e\) and \(\gamma_e\) are assigned with their center vertices receiving value \(2\), their three private leaves receiving value \(0\), and their intermediate vertices receiving value \(0\). One checks directly that every zero vertex has exactly one neighbour of value \(2\). Thus, \(f\) is a perfect Roman dominating function.
	
	The vertex gadgets contribute \(2n\). For each edge \(e=uv\), the two incidence penalty gadgets contribute \(4\), plus one extra unit for each endpoint of \(e\) that belongs to \(D\). The four core vertices \(\alpha_e,\beta_e,\gamma_e,\delta_e\) contribute \(3\) if neither endpoint is in \(D\), and \(2\) if exactly one endpoint is in \(D\). The two penalty gadgets attached to \(\alpha_e\) and \(\gamma_e\) contribute \(4\). Since \(D\) is independent, the extra incidence cost is exactly cancelled by the saving in the core edge gadget. Hence each original edge contributes \(11\), and \(w(f)=2n+11m=K\).
	
	Conversely, suppose that \(H\) has a perfect Roman dominating function \(f\) of weight at most \(K\). For every original vertex \(v\), the triple \(\{v,p_v,q_v\}\) has weight at least \(2\). Also, a penalty gadget attached to a support vertex \(s\) has weight at least \(2\) if \(f(s)\neq 2\), and at least \(3\) if \(f(s)=2\). For an original edge \(e=uv\), let \(s_u=1\) if \(f(u)=2\), and let \(s_v=1\) if \(f(v)=2\). The two incidence penalty gadgets attached to \(u\) and \(v\) contribute at least \(4+s_u+s_v\). The four core vertices \(\alpha_e,\beta_e,\gamma_e,\delta_e\), together with the two penalty gadgets attached to \(\alpha_e\) and \(\gamma_e\), contribute at least \(7-s_u-s_v+s_us_v\). This is a direct four-case check: the minimum local cost is \(7\) when neither endpoint is assigned value \(2\), \(6\) when exactly one endpoint is assigned value \(2\), and \(6\) when both endpoints are assigned value \(2\).
	
	Therefore each original edge \(e=uv\) contributes at least \(11+s_us_v\). Summing over all edges and adding the vertex-gadget contribution, we get \(w(f)\geq 2n+11m+\ell\), where \(\ell\) is the number of original edges \(uv\in E(G)\) with \(f(u)=f(v)=2\). Since \(w(f)\leq K=2n+11m\), we have \(\ell=0\). Thus, no two adjacent original vertices are both assigned value \(2\).
	
	Moreover, equality forces every vertex gadget to have weight exactly \(2\). Hence each original vertex \(v\) is in one of two forms: either \(f(v)=2\) and \(f(p_v)=f(q_v)=0\), or \(f(v)=0\) and \(f(p_v)=f(q_v)=1\). Let \(D=\{v\in V(G):f(v)=2\}\). Since no two adjacent original vertices are both assigned value \(2\), \(D\) is independent. If \(v\notin D\), then \(f(v)=0\), its private pendants have value \(1\), and equality in the auxiliary gadgets implies that no auxiliary neighbour of \(v\) is assigned value \(2\). Therefore the only possible value-\(2\) neighbors of \(v\) are original neighbours of \(v\) in \(G\). Since \(f\) is a perfect Roman dominating function, \(v\) has exactly one such neighbour. Hence, every vertex outside \(D\) has exactly one neighbour in \(D\), and every vertex in \(D\) is dominated exactly once by itself. Therefore, \(D\) is an efficient dominating set of \(G\).
	
	Thus, \(G\) has an efficient dominating set if and only if \(H\) has a perfect Roman dominating function of weight at most \(K\). The construction is polynomial and preserves chordal bipartiteness. Hence \textsc{Perfect Roman Domination} is NP-complete on chordal bipartite graphs.
\end{proof}

\section{Conclusion}\label{sec:conclusion}

In this paper, we studied four Roman-type domination problems on two important subclasses of bipartite graphs. For convex bipartite graphs, we developed a unified dynamic programming framework for Roman-\(\{2\}\) domination, double Roman domination, perfect Roman domination, and unique response Roman domination. The algorithms exploit a convex ordering of one bipartition class and process the graph from left to right. The main structural idea is that, after each processing step, the unfinished requirements of the processed vertices can be represented by a constant number of boundary indices. This compact representation yields \(O(n^6)\)-time algorithms for all four parameters.

We also established hardness results for chordal bipartite graphs. In particular, we proved that Roman-\(\{2\}\) domination, perfect Roman domination, and unique response Roman domination are NP-complete on this graph class. These results show that excluding induced cycles of length at least six is not sufficient to make these Roman-type domination problems tractable. Consequently, our results reveal a clear algorithmic contrast between convex bipartite graphs, where the interval structure supports polynomial-time dynamic programming, and the broader class of chordal bipartite graphs, where several of the corresponding problems remain computationally hard.

Several questions remain open. First, it would be interesting to improve the \(O(n^6)\) running times of the convex-bipartite algorithms, possibly by reducing the number of state coordinates, eliminating unreachable or equivalent states, or exploiting additional monotonicity in the dynamic programming tables. Since biconvex bipartite graphs and bipartite permutation graphs are subclasses of convex bipartite graphs, the algorithms presented in this paper already apply to them. Nevertheless, their stronger ordering properties may permit substantially faster algorithms.

Second, the complexity of these Roman-type domination problems deserves further investigation on graph classes that properly extend, or are not contained in, the class of convex bipartite graphs, such as circular-convex bipartite graphs, triad-convex bipartite graphs, and other restricted subclasses of tree-convex bipartite graphs.  More generally, the results suggest the broader problem of identifying the precise structural boundary between bipartite graph classes on which Roman-type domination problems are polynomial-time solvable and those on which they remain NP-complete.

\bibliographystyle{plain}
\bibliography{RTD_ALGO_ref}

\end{document}